\makeatletter
  \providecommand*\input@path{}
  \newcommand\addinputpath[1]{
  \expandafter\def\expandafter\input@path
  \expandafter{\input@path{#1}}}
  \addinputpath{importfile/}
  \newcommand{\Rmnum}[1]{\expandafter\@slowromancap\romannumeral #1@}
\makeatother
\RequirePackage{fix-cm}
\documentclass[smallcondensed]{svjour3}     
\smartqed  
\usepackage{graphicx}
\usepackage{amsmath}
\usepackage{amssymb}
\usepackage{arydshln}
\usepackage{float}
\usepackage{cite}
\usepackage{bm}
\usepackage[misc]{ifsym}
\usepackage[justification=centering]{caption}
\usepackage[colorlinks, linkcolor=blue, anchorcolor=blue, citecolor=blue]{hyperref}

\newcommand{\rank}{\textrm{Rank}}
\hypersetup{urlcolor=blue}

 \journalname{Noname}
\begin{document}

\title{Two modifications for Loidreau's code-based cryptosystem
}


\author{Wenshuo Guo         \and
        Fang-Wei Fu 
}


\institute{ \Letter\ Wenshuo Guo \at
              Chern Institute of Mathematics and LPMC, Nankai University, Tianjin 300071, China  \\
              \email{ws\_guo@mail.nankai.edu.cn}           
           \and
           Fang-Wei Fu \at
             Chern Institute of Mathematics and LPMC, Nankai University, Tianjin 300071, China\\
              \email{fwfu@nankai.edu.cn} 
}

\date{Received: date / Accepted: date}

\maketitle

\begin{abstract}
This paper presents two modifications for Loidreau's code-based cryptosystem. Loidreau's cryptosystem is a rank metric code-based cryptosystem constructed by using Gabidulin codes in the McEliece setting. Recently a polynomial-time key recovery attack was proposed to break Loidreau's cryptosystem in some cases. To prevent this attack, we propose the use of subcodes to disguise the secret codes in Modification \Rmnum{1}. In Modification \Rmnum{2}, we choose a random matrix of low column rank over $\mathbb{F}_q$ to mix with the secret matrix. According to our analysis, these two modifications can both resist the existing structural attacks. Additionally, we adopt the systematic generator matrix of the public code to make a reduction in the public-key size. In additon to stronger resistance against structural attacks and more compact representation of public keys, our modifications also have larger information transmission rates.

\end{abstract}

\section{Introduction}\label{section1}
In 1978, McEliece proposed the first code-based public-key cryptosystem, namely the well-known  McEliece cryptosystem based on Goppa codes\cite{McEliece1978public}. Since then cryptologists have made extensive study on its security\cite{lee1988observation,canteaut1998cryptanalysis,kobara2001seman,faugere2016structural}. Apart from some weak keys\cite{loidreau2001weak}, the McEliece cryptosystem still remains secure in general cases. The main drawback of this cryptosystem lies in its large public-key size, which makes it unpractical in many situations. To overcome this problem, many variants have been proposed. In 1986, Niederreiter \cite{Niederreiter1986Knapsack} introduced a knapsack-type cryptosystem using GRS codes, which was shown to be insecure by Sidelnikov in\cite{Sidelnikov1992On}. But if we use Goppa codes in the Niederreiter setting, it was proved to be equivalent to the McEliece cryptosystem in terms of security\cite{xinmei1994on}. GRS codes allow us to reduce the public-key size due to their optimal error-correcting capability. Many variants based on GRS codes were proposed after Niederreiter's work. However, nearly all of these variants were broken one after another because of GRS codes being highly structured. In the variant\cite{Baldi2016Enhanced}, the BBCRS cryptosystem, the authors proposed the use of a dense matrix rather than a permutation matrix to disguise the structure of the underlying GRS code. In this proposal, the column scrambler is a matrix of the form $(R+T)^{-1}$, where $T$ is a sparse matrix and $R$ is a dense matrix of low rank. With this approach, the public code seems quite different from GRS codes. This variant therefore can resist some known structural attacks, such as the Sidelnikov-Shestakov attack \cite{Sidelnikov1992On}. However, in\cite{Couvreur2014Distinguisher} the authors presented a polynomial-time key recovery attack against this variant in some cases. Although we can adjust the parameters to prevent such an attack, it would bring some other problems such as the decryption complexity increasing exponentially and a higher request of error-correcting capability for the underlying code.

In 1985 Gabiduin \cite{gabidulin1985theory} introduced a new family of rank metric codes, known as the Gabidulin codes. Since the complexity of decoding general rank metric codes is much higher than that of decoding Hamming metric codes\cite{chabaud1996crypt,ourivski2002new}, it is feasible to obtain much smaller public-key sizes by building cryptosystems in the rank metric. In\cite{gabidulin1991ideals} the authors proposed to use Gabidulin codes in the McEliece setting and introduced the GPT cryptosystem. Unfortunately, several structural attacks were put forward to completely break this system\cite{Gibson1996The,2008Structural,horlemann2018extension}.To prevent these attacks, variants based on different masking skills for Gabidulin codes were proposed\cite{gabidulin2008attacks,gabidulin2009improving,rashwan2011security,loidreau2010designing,rashwan2010smart}. But in\cite{Otmani2018Improved} the authors declare the failure of all the previous masking techniques for Gabidulin codes. In \cite{faure2005new} Faure and Loidreau proposed a cryptosystem also relying on the Gabidulin codes but not in the McEliece setting. Until the work in \cite{gaborit2018polynomial}, the Faure-Loidreau system had never been severely attacked. Recently, in \cite{loidreau2017new} Loidreau proposed a cryptosystem constructed by using Gabidulin codes in the McEliece setting. Different from the original GPT cryptosystem, the isometric matrix is replaced with a matrix whose inverse is taken in an $\mathbb{F}_q$-subspace of $\mathbb{F}_{q^m}$ of dimension $\lambda$. By doing this, the public code seems quite random. Loidreau claimed that his proposal could prevent the existing structural attacks. However, this claim was proved to be invalid by the authors in \cite{coggia2020security} when $\lambda=2$ and the code rate is greater than $1/2$. Soon after this, the author in \cite{2020Extending} generalized this attack to the case of $\lambda>2$ and the code rate greater than $1-\frac{1}{\lambda}$. However, it is feasible to prevent this attack even when the secret code rate is greater than $1-\frac{1}{\lambda}$ according to our analysis.

The rest of this paper is organised as follows. In Section \ref{section2} notations and some concepts about rank metric codes used throughout this paper are given. Section \ref{section3} is devoted to a simple descripton of Loidreau's cryptosystem. In Section \ref{section4} we shall introduce part of the Coggia-Couvreur attack (please refer to \cite{coggia2020security} for more details). Following this, our two modifications for Loidreau's cryptosystem will be introduced in Section \ref{section5}, then security analysis of our modifications will be given  in Section \ref{section6}. In Section \ref{section7}, we will give some suggested parameters for different security levels and make a comparison with Loidreau's original scheme in Table \ref{table1} and with some NIST-PQC submissions in Table \ref{table2}. Section \ref{section8} is our conclusion.

\section{Preliminaries}\label{section2}
\subsection{Notations and basic concepts}
Let $q$ be a prime power. Denote by $\mathbb{F}_q$ the finite field with $q$ elements, and $\mathbb{F}_{q^m}$ an extension field of $\mathbb{F}_q$ of degree $m$. For two positive integers $k$ and $n$, denote by $\mathcal{M}_{k,n}(\mathbb{F}_{q^m})$ the set of all $k\times n$ matrices over $\mathbb{F}_{q^m}$, and by $GL_n(\mathbb{F}_{q^m})$ the set of all $n\times n$ invertible matrices over $\mathbb{F}_{q^m}$. For a matrix $M\in\mathcal{M}_{k,n}(\mathbb{F}_{q^m})$, the column rank of $M$ with respect to $\mathbb{F}_q$, denoted by $\textnormal{Clr}_q(M)$, is the largest number of columns of $M$ linearly independent over $\mathbb{F}_q$. Denote by $\langle M\rangle$ the vector space spanned by rows of $M$ over $\mathbb{F}_{q^m}$. 

An $[n,k]$ linear code $\mathcal{C}$ over $\mathbb{F}_{q^m}$ is a $k$-dimensional subspace of $\mathbb{F}_{q^m}^n$. The dual code of $\mathcal{C}$, denoted by $\mathcal{C}^\perp$, is the orthogonal space of $\mathcal{C}$ under the usual Euclidean inner product over $\mathbb{F}_{q^m}$. A $k\times n$ full-rank matrix $G\in\mathcal{M}_{k,n}(\mathbb{F}_{q^m})$ is called a generator matrix of $\mathcal{C}$  if the vector space $\langle G\rangle$ is exactly the code $\mathcal{C}$. A generator matrix of $\mathcal{C}^\perp$ is called a parity-check matrix of $\mathcal{C}$.

\subsection{Rank metric codes}
Now we recall some basic concepts for rank metric and rank metric codes.
\begin{definition}
For a vector $\bm{x}=(x_1,\cdots,x_n)\in\mathbb{F}_{q^m}^n$, the support of $\bm{x}$ denoted by Supp$(\bm{x})$, is defined to be the linear space spanned by coordinates of $\bm{x}$ over $\mathbb{F}_q$. Formally we have
\[\textnormal{Supp}(\bm{x})=\left\{\sum_{i=1}^n\lambda_ix_i:\lambda_i\in\mathbb{F}_q,1\leqslant i\leqslant n\right\}.\]
\end{definition}
 
\begin{definition}
For a vector $\bm{x}\in\mathbb{F}_{q^m}^n$, the rank weight of $\bm{x}$ denoted by $w_R(\bm{x})$, is defined to be the dimension of Supp$(\bm{x})$ over $\mathbb{F}_q$. 
\end{definition}

Given two vectors $\bm{x},\bm{y}\in\mathbb{F}_{q^m}^n$, the rank distance between $\bm{x}$ and $\bm{y}$, denoted by $d_R(\bm{x},\bm{y})$, is defined to be the rank weight of $\bm{x}-\bm{y}$. It is easy to verify that the function $d_R(\cdot,\cdot)$ defines a proper metric on $\mathbb{F}_{q^m}^n$. A code endowed with the rank metric is called a rank metric code, and in this paper by rank metric codes we always mean linear rank metric codes.

\begin{definition}
For a rank metric code $\mathcal{C}\subseteq \mathbb{F}_{q^m}^n$, the minimum rank distance of $\mathcal{C}$, denoted by $d(\mathcal{C})$, is defined as 
\[d(\mathcal{C})=\min\{d_R(\bm{x},\bm{y}):\bm{x},\bm{y}\in\mathcal{C}\ \textnormal{and }\bm{x}\neq\bm{y}\}.\]
\end{definition}

It is easy to verify that the minimum rank (Hamming) distance of a linear code is equal to its minimum rank (Hamming) weight. In the context of Hamming metric codes, the minimum distance $d$ of an $[n,k]$ linear code satisfies the Singleton bound $d\leqslant n-k+1$\cite{ling2004coding}. Similarly, the minimum rank distance of a rank metric code $\mathcal{C}$ satisfies the following Singleton-style bound.
\begin{theorem}[Singleton-style bound]\cite{gabidulin2003reducible}
Let $\mathcal{C}\subseteq\mathbb{F}_{q^m}^n$ be an $[n,k]$ rank metric code, then the minimum rank distance of $\mathcal{C}$ with respect to $\mathbb{F}_q$ satisfies the following inequality
\[d(\mathcal{C})\leqslant n-k+1.\]
\end{theorem}

\begin{remark}
A linear code attaining the Singleton-style bound is called a Maximum Rank Distance (MRD) code. Apparently an $[n,k]$ MRD code can correct up to $\lfloor\frac{n-k}{2}\rfloor$ rank errors.
\end{remark}

The following proposition implies that the maximum rank weight of a rank metric code is bounded from above by the column rank of its generator matrix.

\begin{proposition}\label{proposition4}
For a matrix $M\in\mathcal{M}_{k,n}(\mathbb{F}_{q^m})$ with $\textnormal{Clr}_q(M)=r$, the maximum rank weight of the code $\langle M\rangle$ is bounded by $r$ from above.
\end{proposition}

\begin{proof}
It suffices to prove that for any $\bm{v}\in\langle M\rangle$, we have $w_R(\bm{v})\leqslant r$. Since $\textnormal{Clr}_q(M)=r$, there exists $Q\in GL_n(\mathbb{F}_q)$ such that $MQ=[M'|O]$, where $M'\in\mathcal{M}_{k,r}(\mathbb{F}_{q^m})$ with $\textnormal{Clr}_q(M')=r$ and $O$ is a zero matrix. For any $\bm{v}\in\langle M\rangle$, there exists $\bm{x}\in\mathbb{F}_{q^m}^k$ such that $\bm{v}=\bm{x}M$ and 
\[\bm{v}Q=\bm{x}MQ=\bm{x}[M'|O]=(\bm{x}',\bm{0}),\] 
where $\bm{x}'\in\mathbb{F}_{q^m}^r$ and $\bm{0}$ is a zero vector. Hence we have $w_R(\bm{v})=w_R(\bm{v}Q)\leqslant r$. This concludes the proof.
\end{proof}

\subsection{Gabidulin codes}
Gabidulin codes can be viewed as an analogue of GRS codes in the rank metric setting, and these two types of codes resemble each other closely in the construction principle. GRS codes admit generator matrices with the Vandermonde structure, while Gabidulin codes can be described by Moore matrices defined as follows.

\begin{definition}
For an integer $s$, denote by $[s]$ the $s$-th Frobenius power $q^s$. A matrix $G\in\mathcal{M}_{k,n}(\mathbb{F}_{q^m})$ is called a Moore matrix generated by $\bm{a}=(a_1,\cdots,a_n)\in\mathbb{F}_{q^m}^n$ if the $s$-th row of $G$ equals the coordinate-wise Frobenius power $\bm{a}^{[s-1]}=(a_1^{[s-1]},\cdots,a_n^{[s-1]})$ for each $1\leqslant s\leqslant k$. Formally we have
\begin{align}\label{moore}
G=
\begin{pmatrix}
a_1&a_2&\cdots&a_n\\
a_1^{[1]}&a_2^{[1]}&\cdots&a_n^{[1]}\\
\vdots&\vdots&&\vdots\\
a_1^{[k-1]}&a_2^{[k-1]}&\cdots&a_n^{[k-1]}
\end{pmatrix}.
\end{align}
\end{definition}

For a matrix $G\in\mathcal{M}_{k,n}(\mathbb{F}_{q^m})$, we define $G^{[s]}=(G_{ij}^{[s]})$. For a set $S\subseteq \mathbb{F}_{q^m}^n$, we define $S^{[s]}=\{\bm{x}^{[s]}:\bm{x}\in S\}$. For a linear code $\mathcal{C}\subseteq\mathbb{F}_{q^m}^n$, it is easy to verify that $\mathcal{C}^{[s]}$ is also an $\mathbb{F}_{q^m}$-linear code.

\begin{definition}[Gabidulin codes]\label{definition1}
For a vector $\bm{a}\in\mathbb{F}_{q^m}^n$ with $w_R(\bm{a})=n\leqslant m$, let $G$ be the $k\times n$ Moore matrix generated by $\bm{a}$. The $[n,k]$ Gabidulin code $\mathcal{G}_{n,k}(\bm{a})$ over $\mathbb{F}_{q^m}$ generated by $\bm{a}$ is defined to be the linear space $\langle G\rangle$, namely we have $\mathcal{G}_{n,k}(\bm{a})=\langle G\rangle$. 
\end{definition}

A major reason for Gabidulin codes being widely used in the design of cryptosystems consists in their remarkable error-correcting capability and simple algebraic structure. Now we recall some properties of Gabidulin codes through the following two theorems without proving.

\begin{theorem}\cite{horlemann2015new}\label{horlemann2015new}
The Gabidulin code $\mathcal{G}_{n,k}(\bm{a})$ is an MRD code. In other words, $\mathcal{G}_{n,k}(\bm{a})$ attains the Singleton-style bound for rank metric codes.
\end{theorem}

According to Theorem \ref{horlemann2015new}, the minimum rank distance of $\mathcal{G}_{n,k}(\bm{a})$ is $n-k+1$. This implies that any $\lfloor\frac{n-k}{2}\rfloor$ rank errors can be corrected. In fact, several efficient docoding algorithms for Gabidulin codes already exist (for instance \cite{gabidulin1985theory,loidreau2005welch,richter2004error}).

\begin{theorem}\cite{gaborit2018polynomial}\label{dualcode}
The dual code of $\mathcal{G}_{n,k}(\bm{a})$ is the Gabidulin code $\mathcal{G}_{n,n-k}(\bm{b}^{[k-n+1)]})$ for some $\bm{b}\in\mathcal{G}_{n,n-1}(\bm{a})^\perp$ with $w_R(\bm{b})=n$.
\end{theorem}

\section{Loidreau's scheme}\label{section3}
For a vector $\bm{a}\in \mathbb{F}_{q^m}^n$ with $w_R(\bm{a})=n$, denote by $G$ a generator matrix of $\mathcal{G}_{n,k}(\bm{a})$. For a positive integer $\lambda\ll m$, let $\mathcal{V}\subseteq \mathbb{F}_{q^m}$ be an $\mathbb{F}_q$-linear space of dimension $\lambda$. Now we give a simple description of Loidreau's scheme through the following three algorithms.
\begin{itemize}
\item Key Generation
\item[]Randomly choose $P\in GL_n(\mathbb{F}_{q^m})$ whose entries are taken from $\mathcal{V}$ and compute $G_{pub}=GP^{-1}$. We publish $(G_{pub},t)$ as the public key where $t=\lfloor \frac{n-k}{2\lambda}\rfloor$, and keep $(\bm{a},P)$ as the secret key.
\item Encryption
\item[]For a plaintext $\bm{m}\in\mathbb{F}_{q^m}^k$, randomly choose a vector $\bm{e}\in\mathbb{F}_{q^m}^n$ with $w_R(\bm{e})=t$. The ciphertext corresponding to $\bm{m}$ is computed as $\bm{c}=\bm{m}G_{pub}+\bm{e}$.
\item Decryption
\item[]Compute $\bm{c}'=\bm{c}P=\bm{m}G+\bm{e}P$. Since $w_R(\bm{e}P)\leqslant w_R(\bm{e})\cdot \dim_q(\mathcal{V})\leqslant \lfloor \frac{n-k}{2}\rfloor$, decoding $\bm{c}'$ will lead to the plaintext $\bm{m}$.
\end{itemize}

\section{The Coggia-Couvreur attack}\label{section4}
Before describing the Coggia-Couvreur attack, we first introduce a distinguisher for Gabidulin codes. This distinguisher provides us with a method of distinguishing Gabidulin codes from general ones. 

\subsection{The distinguisher for Gabidulin codes}
Most of cryptosystems based on Gabidulin codes have been proved to be insecure against structural attacks. Although these attacks were proposed to cryptanalyze different variants of the GPT cryptosystem, the principle for their work is based on the same observation that one can distinguish Gabidulin codes from general ones by performing a simple operation on these codes. 

Given a random linear code $\mathcal{C}\subseteq \mathbb{F}_{q^m}^n$ of dimension $k\leqslant n/2$, the expected dimension of the code $\mathcal{C}+\mathcal{C}^{[1]}$ equals $2k$, or equivalently $\mathcal{C}\cap \mathcal{C}^{[1]}=\{\bm{0}\}$ holds with high probability.
But for a Gabidulin code $\mathcal{G}_{n,k}(\bm{a})$, we have $\mathcal{G}_{n,k}(\bm{a})+\mathcal{G}_{n,k}(\bm{a})^{[1]}=\mathcal{G}_{n,k+1}(\bm{a})$, namely the dimension of $\mathcal{G}_{n,k}(\bm{a})+\mathcal{G}_{n,k}(\bm{a})^{[1]}$ is $k+1$. More generally, we have the following two propositions.

\begin{proposition}\cite{coggia2020security}\label{proposition1}
Let $\mathcal{C}\subseteq\mathbb{F}_{q^m}^n$ be a random linear code of length $n$ and dimension $k$. For a non-negative integer $l$ and a positive integer $s<k$, we have
\[\textnormal{Pr}\left(\dim(\mathcal{C}+\mathcal{C}^{[1]}+\cdots+\mathcal{C}^{[s]})\leqslant\min\{n,(s+1)k\}-l\right)=O(q^{-ml}).\]
\end{proposition}

\begin{proposition}\cite{coggia2020security}\label{proposition2}
Let $k\leqslant n$ and $s$ be a positive integer, then for any $\bm{a}\in\mathbb{F}_{q^m}^n$ with $w_R(\bm{a})=n$, we have
\begin{align*}
\mathcal{G}_{n,k}(\bm{a})\cap\mathcal{G}_{n,k}(\bm{a})^{[1]}&=\mathcal{G}_{n,k-1}(\bm{a}^{[1]});\\
\mathcal{G}_{n,k}(\bm{a})+\mathcal{G}_{n,k}(\bm{a})^{[1]}+\cdots+\mathcal{G}_{n,k}(\bm{a})^{[s]}&=\mathcal{G}_{n,k+s}(\bm{a}).
\end{align*}
\end{proposition}

\subsection{Description of the Coggia-Couvreur attack}\label{section4.2}
In this part we investigate the structural vulnerability of Loidreau's cryptosystem in the case of $\lambda=2$ and the dimension of the public code $\mathcal{C}_{pub}=\langle G_{pub}\rangle$ being greater than $n/2$. The principle for the Coggia-Couvreur attack lies in Propositions \ref{proposition1} and \ref{proposition2}. Instead of directly 
operating the public code, the authors in \cite{coggia2020security} consider the dual of the public code because of the following lemma.

\begin{lemma}\cite{coggia2020security}
Any parity-check matrix $H_{pub}$ of $\mathcal{C}_{pub}$ can be expressed as 
\[H_{pub}=H_{sec}P^T,\]
where $H_{sec}$ is a parity-check matrix of the secret Gabidulin code $\mathcal{G}_{n,k}(\bm{a})$.
\end{lemma}

The authors considered the case of $\lambda=2$, namely the linear space $\mathcal{V}\subseteq \mathbb{F}_{q^m}$ has dimension 2 over $\mathbb{F}_q$. Suppose $\mathcal{V}$ is spanned by $\alpha,\beta\in\mathbb{F}_{q^m}^*$ over $\mathbb{F}_q$, namely $\mathcal{V}=\langle \alpha,\beta\rangle_{\mathbb{F}_q}$. Let $H'_{sec}=\alpha H_{sec}$ and $P'=\alpha^{-1}P$, apparently we have $H_{pub}=H'_{sec}P'^T$. It is easy to see that $H_{sec}'$ spans the same code as $H_{sec}$ and entries of $P'$ are contained in $\mathcal{V}'=\langle 1,\alpha^{-1}\beta\rangle_{\mathbb{F}_q}$. Hence it is reasonable to suppose that $\mathcal{V}=\langle 1,\gamma\rangle_{\mathbb{F}_q}$ for some $\gamma\in\mathbb{F}_{q^m}^*$. In this situation, we can express $P^T$ in the form of
\[P^T=P_0+\gamma P_1,\]
where $P_0,P_1\in \mathcal{M}_{n,n}(\mathbb{F}_q)$.

According to Theorem \ref{dualcode}, there exists some $\bm{b}\in\mathcal{G}_{n,n-1}(\bm{a})^\perp$ with $w_R(\bm{b})=n$ such that $\mathcal{G}_{n,k}(\bm{a})^\perp=\mathcal{G}_{n,n-k}(\bm{b})$. We define
\[\bm{g}=\bm{b}P_0, \bm{h}=\bm{b}P_1.\]
As for the triple $(\gamma,\bm{g},\bm{h})$, the authors made the following two assumptions:
\begin{itemize}
\item[(1)]$\mathcal{G}_{n,n-k+2}(\bm{g})\cap\mathcal{G}_{n,n-k+2}(\bm{h})=\{\bm{0}\}$ and $w_R(\bm{g}),w_R(\bm{h})\geq n-k+2$;
\item[(2)]$m>2$ and $\gamma$ is not contained in any proper subfield of $\mathbb{F}_{q^m}$.
\end{itemize}

The rationality for these two assumptions can be explained as follows. According to the authors' experiments on Magma, Assumption (1) holds with an extremely high probability. Apparently $m>2$ is reasonable because of $m\geqslant n$. On the other hand, if $\gamma$ is contained in some proper subfield of $\mathbb{F}_{q^m}$, then the adversary can find $\gamma$ through the exhausting method for the reason that even the union of all proper subfields of $\mathbb{F}_{q^m}$ contains much less elements than $\mathbb{F}_{q^m}$. Hence $\gamma$ cannot be contained in any proper subfield of $\mathbb{F}_{q^m}$.

The core of the Coggia-Couvreur attack is to find the triple $(\gamma,\bm{g},\bm{h})$ or one of its equivalent forms (see \cite{coggia2020security} for more details). With the knowledge of the triple $(\gamma,\bm{g},\bm{h})$ or one of its equivalent forms, one can decrypt any valid ciphertext in polynomial time and hence completely break Loidreau's cryptosystem.

What follows are two lemmas that will be useful for analysing the security of our modifications. For the remaining part of the Coggia-Couvreur attack,  interested readers can refer to \cite{coggia2020security} for more details. Now we introduce these two lemmas without proving. 
\begin{lemma}\cite{coggia2020security}
The code $\mathcal{C}_{pub}^\perp$ is spanned by 
\begin{align}\label{expression2}
\bm{g}+\gamma\bm{h},\bm{g}^{[1]}+\gamma\bm{h}^{[1]},\cdots,\bm{g}^{[n-k-1]}+\gamma\bm{h}^{[n-k-1]}.
\end{align}
\end{lemma}

\begin{lemma}\cite{coggia2020security}\label{lemma4}
Under Assumption (1), we have that $\mathcal{C}_{pub}^\perp+{\mathcal{C}_{pub}^\perp}^{[1]}$ is spanned by 
\[ \bm{g}+\gamma\bm{h}\textnormal{ and }\bm{g}^{[1]},\bm{h}^{[1]},\cdots,\bm{g}^{[n-k-1]},\bm{h}^{[n-k-1]}\textnormal{ and } \bm{g}^{[n-k]}+\gamma^{[1]}\bm{h}^{[n-k]},\]
and
\[(\mathcal{C}_{pub}^\perp+{\mathcal{C}_{pub}^\perp}^{[1]})\cap({\mathcal{C}_{pub}^\perp}^{[1]}+{\mathcal{C}_{pub}^\perp}^{[2]})\]
is spanned by 
\[\bm{g}^{[1]}+\gamma^{[1]}\bm{h}^{[1]}\ \textnormal{and}\ \bm{g}^{[2]},\bm{h}^{[2]},\cdots,\bm{g}^{[n-k-1]},\bm{h}^{[n-k-1]}\ \textnormal{and}\ \bm{g}^{[n-k]}+\gamma^{[1]}\bm{h}^{[n-k]}.\]
\end{lemma}

\begin{remark}
Similar to Lemma \ref{lemma4}, it is easy to verify that
\begin{align}\label{equation1}
(\mathcal{C}_{pub}^\perp+{\mathcal{C}_{pub}^\perp}^{[1]})\cap({\mathcal{C}_{pub}^\perp}^{[1]}+{\mathcal{C}_{pub}^\perp}^{[2]})\cap\cdots\cap({\mathcal{C}_{pub}^\perp}^{[n-k-1]}+{\mathcal{C}_{pub}^\perp}^{[n-k]})
\end{align}
yields a code spanned by 
\begin{align}\label{equation2}
\bm{g}^{[n-k-1]}+\gamma^{[n-k-1]}\bm{h}^{[n-k-1]}\ \textnormal{and}\ \bm{g}^{[n-k]}+\gamma^{[1]}\bm{h}^{[n-k]}.
\end{align}
\end{remark}

The key point for the Coggia-Couvreur attack is that one can obtain (\ref{equation2}) by computing (\ref{equation1}). But if ${\mathcal{C}_{pub}^\perp}^{[i]}+{\mathcal{C}_{pub}^\perp}^{[i+1]}(0\leqslant i\leqslant n-k-1)$ happens to be the whole space $\mathbb{F}_{q^m}^n$, computing (\ref{equation2}) will lead to nothing but the whole space itself, which means that the Coggia-Couvreur attack will fail in this situation. Our first modification for Loidreau's cryptosystem is inspired by this observation. On the other hand, if $\mathcal{C}_{pub}^\perp$ does not contain the full code spanned by (\ref{expression2}), then one cannot obtain (\ref{equation2}) from (\ref{equation1}) either even if ${\mathcal{C}_{pub}^\perp}^{[i]}+{\mathcal{C}_{pub}^\perp}^{[i+1]}(0\leqslant i\leqslant n-k-1)$ is not the whole space. Modification \Rmnum{2} is based on this observation and this is really true according to our analysis in Section \ref{section6}.

\section{Our modifications}\label{section5}
In code-based cryptography, randomness is widely used in both the key generation and encryption procedures. In terms of the intersection of a given linear code and a randomly chosen linear space, we have the following proposition.
\begin{proposition}\label{proposition3}
Let $n,k,l$ be positive integers with $k+l<n$. Let $\mathcal{C}\subseteq{\mathbb{F}_{q^m}^n}$ be a linear code of dimension $k$, and $\mathcal{V}$ be a random linear subspace of $\mathbb{F}_{q^m}^n$ of dimension $l$. In terms of the intersection $\mathcal{C}\cap\mathcal{V}$, we have the following inequality
\begin{align*}
\textnormal{Pr}\{\mathcal{C}\cap\mathcal{V}=\{\bm{0}\}\}\geqslant 1-O\big(q^{-ms}\big),
\end{align*} 
where $s\geqslant 2$ is a positive integer.
\end{proposition}

\begin{proof}
Exploiting the Gaussian  coefficient, the number of $l$-dimensional subspaces of $\mathbb{F}_{q^m}^n$ linearly independent of $\mathcal{C}$ can be computed as
\[N_1=
\prod_{i=0}^{l-1}\frac{(q^m)^n-(q^m)^{k+i}}{(q^m)^l-(q^m)^i}=\prod_{i=0}^{l-1}\frac{q^{mn}-q^{m(k+i)}}{q^{ml}-q^{mi}}.
\]
Similarly, the number of all $l$-dimensional subspaces of $\mathbb{F}_{q^m}^n$ can be computed as
\[N_2=
\prod_{i=0}^{l-1}\frac{(q^m)^n-(q^m)^i}{(q^m)^l-(q^m)^i}=\prod_{i=0}^{l-1}\frac{q^{mn}-q^{mi}}{q^{ml}-q^{mi}}.
\]
Then the target probability $\textnormal{Pr}\{\mathcal{C}\cap\mathcal{V}=\{\bm{0}\}\}$ can be computed as 
\begin{align}\label{equation}
\frac{N_1}{N_2}&=\prod_{i=0}^{l-1}\frac{q^{mn}-q^{m(k+i)}}{q^{mn}-q^{mi}}\notag\\
&=\prod_{i=0}^{l-1}\frac{q^{mn}-q^{mi}-q^{mk+mi}+q^{mi}}{q^{mn}-q^{mi}}\notag\\
&=\prod_{i=0}^{l-1}\Big(1-\frac{q^{mk}-1}{q^{m(n-i)}-1}\Big)\notag\\
&\geqslant\Big(1-\frac{q^{mk}-1}{q^{m(n-l+1)}-1}\Big)^l.
\end{align}
By Taylor expansion, the right hand side of (\ref{equation}) can be expressed as
\begin{align*}
\Big(1-\frac{q^{mk}-1}{q^{m(n-l+1)}-1}\Big)^l&=1-l\cdot \frac{q^{mk}-1}{q^{m(n-l+1)}-1}+o\Big(\frac{q^{mk}-1}{q^{m(n-l+1)}-1}\Big)\\
&=1-O\Big(q^{-m(n-k-l+1)}\Big).
\end{align*}
Let $s=n-k-l+1$, apparently $s\geqslant 2$ because of $k+l<n$. Finally we have $\textnormal{Pr}\{\mathcal{C}\cap\mathcal{V}=\{\bm{0}\}\}\geqslant 1-O\big(q^{-ms}\big)$. This completes the proof.
\end{proof}

\begin{remark}
Proposition \ref{proposition3} states a fact that for a linear code $\mathcal{C}$ and a randomly chosen linear space $\mathcal{V}$, we have that $\mathcal{C}\cap\mathcal{V}=\{\bm{0}\}$ holds with high probability. Meanwhile, it is reasonable to conclude that for a $k\times n$ full-rank matrix $H$ and a randomly chosen $l\times n$ full-rank matrix $A$ with $k+l<n$, the block matrix $\begin{pmatrix}A\\H\end{pmatrix}$ is of full rank with high probability.
\end{remark}

\subsection{Description of Modification \Rmnum{1}}
Let $\mathcal{G}$ be an $[n,k]$ Gabidulin code generated by $\bm{a}\in\mathbb{F}_{q^m}^n$ with $w_R(\bm{a})=n$. Denote by $H$ a parity-check matrix of $\mathcal{G}$. For a positive integer $l\geqslant k-\frac{n}{2}$, randomly choose an $l\times n$ full-rank matrix $A$ over $\mathbb{F}_{q^m}$ and set $H_{sub}=\begin{pmatrix}A\\H\end{pmatrix}$. Let $G_{sub}$ be a generator matrix of $\langle H_{sub}\rangle^\perp$. By Proposition \ref{proposition3}, $H_{sub}$ has rank $k+l$ with high probability.  It would be well if we assume that $H_{sub}$ is of full rank, otherwise we rechoose the matrix $A$. Apparently $G_{sub}$ spans a subcode of $\mathcal{G}$ of dimension $k'=k-l$. For a positive integer $\lambda\ll m$, let $\mathcal{V}\subseteq \mathbb{F}_{q^m}$ be an $\mathbb{F}_q$-linear space of dimension $\lambda$.

\begin{itemize}
\item Key generation
\item[] Let $P\in GL_n(\mathbb{F}_{q^m})$ with entries contained in $\mathcal{V}$. Without loss of generality, we assume that the submatrix of $G_{sub}P^{-1}$ formed by the first $k'$ columns is invertible. Choose a matrix $S\in GL_{k'}(\mathbb{F}_{q^m})$ to change $G_{pub}=SG_{sub}P^{-1}$ into systematic form. We publish $(G_{pub},t)$ as the public key where $t=\lfloor\frac{n-k}{2\lambda}\rfloor$, and keep $(\bm{a},P)$ as the secret key.
\item Encryption
\item[] For a plaintext $\bm{m}\in\mathbb{F}_{q^m}^{k'}$, randomly choose $\bm{e}\in\mathbb{F}_{q^m}^n$ with $w_R(\bm{e})=t$. Then the ciphertext corresponding to $\bm{m}$ is computed as $\bm{c}=\bm{m}G_{pub}+\bm{e}$.
\item Decryption
\item[] For a ciphertext $\bm{c}$, compute $\bm{c}'=\bm{c}P=\bm{m}SG_{sub}+\bm{e}P$. Since $w_R(\bm{e}P)\leqslant w_R(\bm{e})\cdot \lambda\leqslant \lfloor\frac{n-k}{2}\rfloor$. Applying the decoding procedure of  $\mathcal{G}$ to $\bm{c}'$ will lead to $\bm{e}'=\bm{e}P$, then we have $\bm{e}=\bm{e}'P^{-1}$. The restriction of $\bm{c}-\bm{e}$ to the first $k'$ coordinates will be the plaintext $\bm{m}$.
\end{itemize}

\begin{remark}\label{remark4}
According to the analysis in Section \ref{section4.2}, we can always assume that $1\in\mathcal{V}$. If $\lambda=1$, there will be $\mathcal{V}=\mathbb{F}_q$ and $P^{-1}\in GL_n(\mathbb{F}_q)$. In this situation, $G_{pub}$ spans a subcode of $\mathcal{G}$. Then one can exploit the $r$-Frobenius weak attack \cite{horlemann2016} to completely break this modification. To prevent this attack, we should make sure that $\lambda\geqslant 2$ in Modification \Rmnum{1}. 
\end{remark}

\subsection{Description of Modification \Rmnum{2}}
Let $\mathcal{G}$ be an $[n,k]$ Gabidulin code generated by $\bm{a}\in\mathbb{F}_{q^m}^n$ with $w_R(\bm{a})=n$. Denote by $G$ a generator matrix of $\mathcal{G}$. For a positive integer $l\ll \min\{k,n-k\}$, randomly choose $M\in\mathcal{M}_{k,n}(\mathbb{F}_{q^m})$ with $\textnormal{Clr}_q(M)=l$ and let $G_M=G+M$. It is easy to see that $G_M$ is of full rank. Indeed, if there exists $\bm{x}\in\mathbb{F}_{q^m}^k$ such that $\bm{x}G_M=\bm{0}$, then we have $\bm{x}G\in\langle M\rangle$. By Proposition \ref{proposition4}, the maximum rank weight of $\langle M\rangle$ does not exceed $l$. Together with $d(\mathcal{G})=n-k+1\gg l$  , we have $\bm{x}G=\bm{0}$ and hence $\bm{x}=\bm{0}$. For a positive integer $\lambda\ll m$, let $\mathcal{V}\subseteq \mathbb{F}_{q^m}$ be an $\mathbb{F}_q$-linear space of dimension $\lambda$.
\begin{itemize}
\item Key generation
\item[]Let $P\in GL_n(\mathbb{F}_{q^m})$ with entries contained in $\mathcal{V}$.  Without loss of generality, we assume that the submatrix of $G_MP^{-1}$ formed by the first $k$ columns is invertible. Choose a matrix $S\in GL_k(\mathbb{F}_{q^m})$ to change $G_{pub}=SG_MP^{-1}$ into systematic form.  We publish $(G_{pub},t)$ as the public key where $t=\lfloor\frac{n-k-2l}{2\lambda}\rfloor$, and keep $(S,G,P)$ as the secret key. 
\item Encryption
\item[]For a plaintext $\bm{m}\in\mathbb{F}_{q^m}^k$, randomly choose a vector $\bm{e}\in\mathbb{F}_{q^m}^n$ with $w_R(\bm{e})=t$. Then the ciphertext corresponding to $\bm{m}$ is computed as $\bm{c}=\bm{m}G_{pub}+\bm{e}$.
\item Decryption
\item[]For a ciphertext $\bm{c}$, compute $\bm{c}'=\bm{c}P=\bm{m}SG+\bm{m}SM+\bm{e}P$. Since 
\begin{align*}
w_R(\bm{m}SM+\bm{e}P)\leqslant w_R(\bm{m}SM)+w_R(\bm{e}P)\leqslant l+\lambda t\leqslant \lfloor\frac{n-k}{2}\rfloor,
\end{align*}
applying the decoding procedure of $\mathcal{G}$ to $\bm{c}'$ will lead to $\bm{m}SG$. Then the plaintext $\bm{m}$ can be recovered by solving a linear system with a complexity of $O(n^3)$.
\end{itemize}

\begin{remark}
Similar to the analysis in Remark \ref{remark4}, we should make sure that $\lambda\geqslant 2$ in this modification. Otherwise, Modification \Rmnum{2} can be reduced to the GPT cryptosystem that has been completely broken.
\end{remark}

\section{Security analysis}\label{section6}
In general, there are two types of attacks on code-based cryptosystems, namely the generic attack and structural attack.

\textbf{Generic attacks}. These attacks aim to recover the plaintext directly from the ciphertext when nothing but the public key is known. In the context of code-based cryptography, generic attacks are involved with the problem of decoding general linear codes or equivalently the syndrome decoding problem, both of which are believed to be very difficult by the community. 
In the paper \cite{gaborit2016complexity}, the authors proposed two generic attacks on the rank syndrome decoding (RSD) proplem, which lay a foundation for the security of rank metric code-based cryptography. 

Suppose $\mathcal{C}$ is an $[n,k]$ rank metric code over $\mathbb{F}_{q^m}$, correcting up to $t$ rank errors. Let $\bm{y}=\bm{c}+\bm{e}$, where $\bm{c}$ is a codeword in $\mathcal{C}$ and $\bm{e}$ is a random vector with $w_R(\bm{e})=t$. The first attack in \cite{gaborit2016complexity} is combinatorial and permits to recover the error $\bm{e}$ with 
\[\min\{O((n-k)^3m^3q^{t\lfloor (km)/ n\rfloor}),O((n-k)^3m^3q^{(t-1)\lfloor ((k+1)m)/ n\rfloor})\}\] operations in $\mathbb{F}_q$. The second attack in \cite{gaborit2016complexity} is algebraic and shows that in the case of $\lceil((t+1)(k+1)-(n+1))/t\rceil\leqslant k$, the RSD problem can be solved with an average complexity of 
\[O(t^3k^3q^{t \lceil((t+1)(k+1)-(n+1))/t\rceil})\]
in $\mathbb{F}_q$. Apparently both of these two attacks need exponential time to recover the plaintext without knowing the secret key. 
 
\textbf{Structural attacks}. These attacks aim to recover the structure of the secret code from a random-looking public matrix. In fact, recovering the structure amounts to obtaining the secret key in some sense, which means that the cryptosystem will be completely broken in this situation. In \cite{loidreau2017new}, Loidreau argued that his cryptosystem could resist the invariant subspace attack, also known as Overbeck's attack. Since our modifications exploit the same masking technique to disguise the structure of the secret code, naturally we believe that our modifications can also prevent Overbeck's attack. Therefore, in the remaining part of this section we only consider the security against the Coggia-Couvreur attack. 

\subsection{Analysis of Modification \Rmnum{1}}
Before giving the analysis, we shall introduce the following theorem. This theorem states a simple fact that if $\mathcal{C}\subseteq\mathbb{F}_{q^m}^m$ is a linear code with a generator matrix $G$, then its $s$-th Frobenius power $\mathcal{C}^{[s]}$ is also a linear code over $\mathbb{F}_{q^m}$ and has $G^{[s]}$ as a generator matrix.
\begin{theorem}\label{theorem2}
Let $\mathcal{C}\subseteq\mathbb{F}_{q^m}^n$ be an $[n,k]$ linear code that has $G$ as a generator matrix. For any integer $s$, $\mathcal{C}^{[s]}$ is also an $[n,k]$ linear code over $\mathbb{F}_{q^m}$ and has $G^{[s]}$ as a generator matrix.
\end{theorem}

\begin{proof}
On the one hand. For any $\bm{u}\in\mathcal{C}^{[s]}$, there exists $\bm{x}\in\mathbb{F}_{q^m}^k$ such that $\bm{u}=(\bm{x}G)^{[s]}=\bm{x}^{[s]}G^{[s]}\in\langle G^{[s]}\rangle$, then we have 
\[\mathcal{C}^{[s]}\subseteq\langle G^{[s]}\rangle.\]
On the other hand. For any $\bm{v}\in\langle G^{[s]}\rangle$, there exists $\bm{x}\in\mathbb{F}_{q^m}^k$ such that $\bm{v}=\bm{x}G^{[s]}=(\bm{x}^{[m-s]}G)^{[s]}\in\mathcal{C}^{[s]}$, then we have 
\[\langle G^{[s]}\rangle\subseteq\mathcal{C}^{[s]}.\]
Hence we have $\mathcal{C}^{[s]}=\langle G^{[s]}\rangle$.   

It remains to prove that $G^{[s]}$ is of full rank. Suppose there exists $\bm{x}\in\mathbb{F}_{q^m}^k$ such that $\bm{x}G^{[s]}=(\bm{x}^{[m-s]}G)^{[s]}=\bm{0}$, then we have $\bm{x}^{[m-s]}G=\bm{0}$ and consequently $\bm{x}=\bm{x}^{[m-s]}=\bm{0}$ because of $G$ being of full rank. This concludes the proof.
\end{proof}

Now we show that ${\mathcal{C}_{pub}^\perp}^{[i]}+{\mathcal{C}_{pub}^\perp}^{[i+1]}\,(0\leqslant i\leqslant n-k-1)$ is exactly the whole space $\mathbb{F}_{q^m}^n$, namely all these $n-k$ codes have dimension $n$. By Theorem \ref{theorem2}, it suffices to consider the case of $\mathcal{C}_{pub}^\perp+{\mathcal{C}_{pub}^\perp}^{[1]}$. 

Let $H_{pub}$ be a parity-check matrix of $\mathcal{C}_{pub}$, then we have $H_{pub}=H_{sub}P^T$ and
\[\mathcal{C}_{pub}^\perp=\langle H_{sub}P^T\rangle=\langle HP^T\rangle+\langle AP^T\rangle.\]
Hence
\[\mathcal{C}_{pub}^\perp+{\mathcal{C}_{pub}^\perp}^{[1]}=\langle HP^T\rangle+\langle HP^T\rangle^{[1]}+\langle AP^T\rangle+\langle AP^T\rangle^{[1]}.\]
According to Lemma \ref{lemma4}, $\langle HP^T\rangle+\langle HP^T\rangle^{[1]}$ is spanned by
\begin{align}\label{expression1}
\bm{g}+\gamma\bm{h}\ \textnormal{and}\ \bm{g}^{[1]},\bm{h}^{[1]},\cdots,\bm{g}^{[n-k-1]},\bm{h}^{[n-k-1]}\ \textnormal{and}\ \bm{g}^{[n-k]}+\gamma^{[1]}\bm{h}^{[n-k]},
\end{align}
where $\gamma,\bm{g}$ and $\bm{h}$ are defined as in Section \ref{section4}.

Note that these $2(n-k)$ vectors in (\ref{expression1}) are linearly independent over $\mathbb{F}_{q^m}$. Indeed, if there exist $x_i,y_i\in\mathbb{F}_{q^m}\,(0\leqslant i\leqslant n-k-1)$ such that 
\[x_0(\bm{g}+\gamma\bm{h})+y_0(\bm{g}^{[n-k]}+\gamma^{[1]}\bm{h}^{[n-k]})+\sum_{i=1}^{n-k-1}x_i\bm{g}^{[i]}+\sum_{i=1}^{n-k-1}y_i\bm{h}^{[i]}=\bm{0}.\]
Then we have
\[y_0\bm{g}^{[n-k]}+\sum_{i=0}^{n-k-1}x_i\bm{g}^{[i]}=-x_0\gamma\bm{h}-y_0\gamma^{[1]}\bm{h}^{[n-k]}-\sum_{i=1}^{n-k-1}y_i\bm{h}^{[i]}.\]
Apparently $y_0\bm{g}^{[n-k]}+\sum_{i=0}^{n-k-1}x_i\bm{g}^{[i]}\in\mathcal{G}_{n,n-k+2}(\bm{g})$ and $-x_0\gamma\bm{h}-y_0\gamma^{[1]}\bm{h}^{[n-k]}-\sum_{i=1}^{n-k-1}y_i\bm{h}^{[i]}\in\mathcal{G}_{n,n-k+2}(\bm{h})$. Hence $x_i=y_i=0\,(0\leqslant i\leqslant n-k-1)$ because of Assumption (1).

By Proposition \ref{proposition1}, we have that $\dim(\langle AP^T\rangle+\langle AP^T\rangle^{[1]})=2l$ holds with extremely high probability. Together with Proposition \ref{proposition3}, we have that $\dim(\mathcal{C}_{pub}^\perp+{\mathcal{C}_{pub}^\perp}^{[1]})=n=\min\{2(n-k+l),n\}$. This means that by computing the intersection (\ref{equation1}) the adversary can obtain nothing but the whole space and hence the Coggia-Couvreur attack will fail in this situation.

\subsection{Analysis of Modification \Rmnum{2}}
Since $\textnormal{Clr}_q(M)=l$, there must be $1\leqslant\rank(M)\leqslant l$. Assume that $\rank(M)=l'$, apparently we have $\dim(\langle M\rangle)=l'\leqslant l$. By Proposition \ref{proposition4}, we have $w_R(\bm{v})\leqslant l$ for any $\bm{v}\in \langle M\rangle$. Together with $d(\mathcal{G})=n-k+1\gg l$, we have $\langle M\rangle\cap\mathcal{G}=\{\bm{0}\}$. 

Let $\mathcal{C}_{pub}=\langle G_{pub}\rangle=\langle SG_MP^{-1}\rangle$, then a parity-check matrix for $\mathcal{C}_{pub}$ can be written as $H_{pub}=H_MP^T$, where $H_M$ is an $(n-k)\times n$ full-rank matrix such that $SG_MH_M^T=O$. It is easy to see that $\langle H_M\rangle$ contains a subcode of $\mathcal{G}^\perp$ of dimension $n-k-l'$. Hence $\mathcal{C}_{pub}^\perp$ contains a subcode of $\mathcal{C}_1$ of dimension $n-k-l'$, where $\mathcal{C}_1$ is spanned by 
\[\bm{g}+\gamma\bm{h},\bm{g}^{[1]}+\gamma\bm{h}^{[1]},\cdots,\bm{g}^{[r]}+\gamma\bm{h}^{[r]},\ \textnormal{where $r=n-k-1$}.\]
Similarly ${\mathcal{C}_{pub}^\perp}^{[1]}$ contains a subcode of $\mathcal{C}_2$ of dimension $n-k-l'$, where $\mathcal{C}_2$ is spanned by 
\[\bm{g}^{[1]}+\gamma^{[1]}\bm{h}^{[1]},\bm{g}^{[2]}+\gamma^{[1]}\bm{h}^{[2]},\cdots,\bm{g}^{[r+1]}+\gamma^{[1]}\bm{h}^{[r+1]}.\] 
Finally we have that $\mathcal{C}_{pub}^\perp+{\mathcal{C}_{pub}^\perp}^{[1]}$ contains a subcode of  $\mathcal{C}=\mathcal{C}_1+\mathcal{C}_2$ of dimension at most $2(n-k-l')$, where $\mathcal{C}$ is spanned by 
\[\bm{g}+\gamma\bm{h}\ \textnormal{and}\ \bm{g}^{[1]},\bm{h}^{[1]},\cdots, \bm{g}^{[r]},\bm{h}^{[r]}\ \textnormal{and}\ \bm{g}^{[r+1]}+\gamma^{[1]}\bm{h}^{[r+1]}.\]

In the Coggia-Couvreur attack, the adversary can obtain (\ref{equation2}) by computing (\ref{equation1}). Our analysis shows that the adversary cannot perform the same operation on Modification \Rmnum{2} to obtain (\ref{equation2}). Here we demonstrate this point with the method of reduction to absurdity.

Suppose that
\begin{align}\label{contradiction}
\langle \bm{g}^{[r]}+\gamma^{[r]}\bm{h}^{[r]},\bm{g}^{[r+1]}+\gamma^{[1]}\bm{h}^{[r+1]}\rangle\subseteq\bigcap_{i=0}^r ({\mathcal{C}_{pub}^\perp}^{[i]}+{\mathcal{C}_{pub}^\perp}^{[i+1]}).
\end{align}
Then for any $0\leqslant i\leqslant r$, we have 
\begin{equation}\label{inverse}
\bm{g}^{[r]}+\gamma^{[r]}\bm{h}^{[r]},\bm{g}^{[r+1]}+\gamma^{[1]}\bm{h}^{[r+1]}\in{\mathcal{C}_{pub}^\perp}^{[i]}+{\mathcal{C}_{pub}^\perp}^{[i+1]}.
\end{equation}
Applying the inverse of the $i$-th Frobenius map to both sides of (\ref{inverse}), there will be
\[\bm{g}^{[r-i]}+\gamma^{[r-i]}\bm{h}^{[r-i]},\bm{g}^{[r-i+1]}+\gamma^{[1-i]}\bm{h}^{[r-i+1]}\in\mathcal{C}_{pub}^\perp+{\mathcal{C}_{pub}^\perp}^{[1]},\]
or equivalently
\[\bm{g}+\gamma\bm{h}\ \textnormal{and}\ \bm{g}^{[1]},\bm{h}^{[1]},\cdots, \bm{g}^{[r]},\bm{h}^{[r]}\ \textnormal{and}\ \bm{g}^{[r+1]}+\gamma^{[1]}\bm{h}^{[r+1]}\in\mathcal{C}_{pub}^\perp+{\mathcal{C}_{pub}^\perp}^{[1]}.\] 
This implies that $\mathcal{C}\subseteq \mathcal{C}_{pub}^\perp+{\mathcal{C}_{pub}^\perp}^{[1]}$, which conflicts with the previous conclusion that $\mathcal{C}_{pub}^\perp+{\mathcal{C}_{pub}^\perp}^{[1]}$ contains a subcode of $\mathcal{C}$ of dimension at most $2(n-k-l')$.

Hence the assumption (\ref{contradiction}) cannot be true and the adversary cannot recover $(\ref{equation2})$ from $(\ref{equation1})$ as the Coggia-Couvreur attack on Loidreau's cryptosystem. Therefore the Coggia-Couvreur attack does not work on Modification \Rmnum{2}.

\section{Parameters and key size}\label{section7}
In Table \ref{table1} we give some parameters suggested for different security levels, and make a comparison on performance with Loidreau's original scheme in the case of $k\leqslant \frac{n}{2}$. When considering the parameters, we exploit the complexity assessment of generic attacks given in Section \ref{section6}. 

In Modification \Rmnum{1}, the public key is a systematic generator matrix of an $[n,k-l]$ rank metric code, resulting in a public-key size of $(k-l)(n-k+l)\cdot m\cdot\log_2(q)$ bits. In Modification \Rmnum{2}, the public key is a systematic generator matrix of an $[n,k]$ rank metric code, resulting in a public-key size of $k(n-k)\cdot m\cdot\log_2(q)$ bits. As for information rates, this value is $(k-l)/n$ for Modification \Rmnum{1}, and $k/n$ for Modification \Rmnum{2} respectively. For the concrete instances, we consider the case where $q=3$ and $\lambda=2$. It is not difficult to see from Table \ref{table1} that our modifications have obvious advantages over Loidreau's scheme in both public-key sizes and information rates.
\begin{table}[h!]
\setlength{\abovecaptionskip}{-0.2cm}
\setlength{\belowcaptionskip}{-0.2cm}
\begin{center}
\begin{tabular}{c|c|c|c|c}
\hline\rule{0pt}{10pt}       
Instance & Parameters & Public-key Size & Inf. Rate & Sec. \\ 
\hline\rule{0pt}{10pt}
& m=37, n=37, k=17 & 4,611 & 0.46 & 128 \\  \rule{0pt}{10pt}
Loidreau's system & m=45, n=45, k=21 & 8,425 & 0.47 & 192 \\   \rule{0pt}{10pt}
& m=52, n=52, k=24 & 12,857 & 0.46 & 256 \\ 
\hline\rule{0pt}{10pt}
& m=42, n=42, k=23, l=2 & 3,670 & 0.50 & 128 \\  \rule{0pt}{10pt}
Modification \Rmnum{1} & m=48, n=48, k=25, l=1 & 5,478 & 0.50 & 192 \\   \rule{0pt}{10pt}
& m=56, n=56, k=29, l=1 & 8,698 & 0.50 & 256 \\
\hline\rule{0pt}{10pt}
& m=44, n=44, k=30, l=1 & 3,661 & 0.68 & 128 \\  \rule{0pt}{10pt}
Modification \Rmnum{2} & m=51, n=51, k=33, l=1 & 6,002 & 0.65 & 192 \\   \rule{0pt}{10pt}
& m=57, n=57, k=35, l=1 & 8,696 & 0.61& 256 \\
\hline
\end{tabular}
\end{center}
\caption{Comparison on public-key sizes (in bytes) and information rates with Loidreau's scheme for different security levels.}\label{table1}
\end{table}

\begin{table}[h!]
\setlength{\abovecaptionskip}{-0.2cm}
\setlength{\belowcaptionskip}{-0.2cm}
\begin{center}
\begin{tabular}{c|r|r|r}
\hline\rule{0pt}{10pt}       
Instance & 128 bits & 192 bits & 256 bits\\
\hline\rule{0pt}{10pt}
HQC & 2,249 & 4,522 & 7,245\\
\hline\rule{0pt}{10pt}
BIKE & 1,540 & 3,082 & 5,121\\
\hline\rule{0pt}{10pt}
Classic McEliece & 261,120 & 524,160 & 1,044,992\\
\hline\rule{0pt}{10pt}
NTS-KEM & 319,488 & 929,760 & 1,419,704\\
\hline\rule{0pt}{10pt}
Modification \Rmnum{1} & 3,693 & 5,478 & 8,698\\
\hline\rule{0pt}{10pt}
Modification \Rmnum{2} & 3,661 & 6,002 & 8,696\\
\hline
\end{tabular}
\end{center}
\caption{Comparison on public-key sizes (in bytes) with some other cryptosystems.}\label{table2}
\end{table}

In Table \ref{table2}, we make a comparison on public-key sizes with some other code-based cryptosystems that were selected as the third round candidates of the NIST PQC Standardization Process. These candidates are HQC\cite{aguilar2020}, BIKE\cite{aragon2020}, NTS-KEM\cite{ntskem2019} and Classic McEliece\cite{daniel2020}. Note that the Classic McEliece published in the third round of the NIST PQC project is a merged version of NTS-KEM and the original Classic McEliece for their specifications being very similar. From Table \ref{table2} we can see that our modifications behave pretty well without using codes endowed with special algebraic structures.

\section{Conclusion}\label{section8}
In this paper, we propose two modifications for Loidreau's cryptosystem. According to our analysis, both of these two modifications can resist the existing structural attacks on Gabidulin codes based cryptosystems, including Overbeck's attack and the Coggia-Couvreur attack. In our modifications, we adopt a systematic generator matrix of the public code to reduce the public-key size. Note that this method of describing the public code may reveal some information about the plaintext because of the sparsity of the intended errors in Hamming metric\cite{canteaut1995improvements}, which means a security flaw to the cryptosystem. In the rank metric, however, the intended errors may happen in all coordinates of the error vector with high probability. Particularly, if we generate the error vector by randomly and uniformly choosing $n$ elements from an $\mathbb{F}_{q}$-subspace of $\mathbb{F}_{q^m}$ of dimension $t$, then the expected Hamming weight of the subvector of length $k$ is $k(1-\frac{1}{q^t})\sim k$, while in Hamming metric this value is $kt/n$. Therefore there is no need to worry about this problem in our modifications.

\begin{acknowledgements}
This research is supported by the National Key Research and Development Program of China (Grant No. 2018YFA0704703), the National Natural Science Foundation of China (Grant No. 61971243), the Natural Science Foundation of Tianjin (20JCZDJC00610), and the Fundamental Research Funds for the Central Universities of China (Nankai University).
\end{acknowledgements}

\end{document}